\documentclass[10pt, twocolumn]{IEEEtran}
\usepackage{subfigure, epsfig, color}
\usepackage{amsmath}
\usepackage{amsthm}
\usepackage{amssymb}
\usepackage{multirow}
\usepackage{booktabs}
\newtheorem{theorem}{Theorem}
\newtheorem{lemma}{Lemma}

\usepackage{mathtools}
\newcommand\myeqa{\stackrel{\mathclap{\normalfont\mbox{(a)}}}{=}}
\newcommand\myeqb{\stackrel{\mathclap{\normalfont\mbox{(b)}}}{=}}

\makeatletter
\newcommand*{\rom}[1]{\expandafter\@slowromancap\romannumeral #1@}
\makeatother
\begin{document}

\title{One-Bit OFDM Receivers via Deep Learning}
\author{{
    Eren Balevi and
    Jeffrey G. Andrews}\\
\thanks{The authors are with the University of Texas at Austin, TX, Email: erenbalevi@utexas.edu, jandrews@ece.utexas.edu.  This work has been supported by Samsung Electronics and NSF grant CCF-1514275.  This paper was presented in part at the Asilomar Conference, Pacific Grove, CA, October 28-31, 2018 \cite{ebalevi}.}				
}
\maketitle 
\normalsize
\begin{abstract}
This paper develops novel deep learning-based architectures and design methodologies for an orthogonal frequency division multiplexing (OFDM) receiver under the constraint of one-bit complex quantization. Single bit quantization greatly reduces complexity and power consumption, but makes accurate channel estimation and data detection difficult. This is particularly true for multicarrier waveforms, which have high peak-to-average power ratio in the time domain and fragile subcarrier orthogonality in the frequency domain. The severe distortion for one-bit quantization typically results in an error floor even at moderately low signal-to-noise-ratio (SNR) such as 5 dB. For channel estimation (using pilots), we design a novel generative supervised deep neural network (DNN) that can be trained with a reasonable number of pilots. After channel estimation, a neural network-based receiver -- specifically, an autoencoder -- jointly learns a precoder and decoder for data symbol detection. Since quantization prevents end-to-end training, we propose a two-step sequential training policy for this model. With synthetic data, our deep learning-based channel estimation can outperform least squares (LS) channel estimation for unquantized (full-resolution) OFDM at average SNRs up to 14 dB. For data detection, our proposed design achieves lower bit error rate (BER) in fading than unquantized OFDM at average SNRs up to 10 dB.  
\end{abstract}


\section{Introduction}
Wireless systems are trending towards ever-higher data rates, which requires ever-more antennas and bandwidth; a canonical example being millimeter wave (mmWave) systems \cite{RapHeaBook}.  Analog-to-digital converters (ADCs) consume a significant fraction of the power in modern receivers \cite{Walden}, which is a key bottleneck to large bandwidth and many antenna systems. One-bit quantization dramatically reduces the power consumption, e.g., by more than two orders of magnitude in some cases \cite{Walden}, and can perform satisfactorily for a large amount of receive antenna combining (which averages the quantization noise) or at low signal-to-noise-ratios (SNRs) \cite{Risi}-\cite{Jacobsson2}. However, one-bit ADCs fundamentally have poor performance at medium and high SNRs \cite{Mo1}, \cite{Mo2} or for the case of few receive antennas. Additionally, OFDM waveforms -- which are the core of the physical layer for virtually all modern high-rate wireless systems -- are more sensitive to one-bit quantization than single carrier systems. This is because OFDM waveforms have a high peak-to-average power ratio, and so one-bit quantization leads to severe inter-carrier interference (ICI) in the frequency domain, where channel estimation and data detection are performed. Yet most prior work has been for single carrier communication.

Recognizing that one-bit quantization introduces strong nonlinearities and other intractable features that render traditional OFDM receiver architectures far from optimal, and motivated by the success of deep learning in many different challenging applications \cite{Mnih}-\cite{Sutskever}, this paper and the design and methodology herein are the result of exploring many possible different neural network architectures for channel estimation and data detection. We consider a single antenna receiver and a moderate (e.g. $64$) number of subcarriers in a frequency selective fading channel and summarize the main contributions in Sect. \ref{sec:cont}.

\subsection{Related Work}
One-bit ADCs have been extensively researched in terms of channel estimation and data detection. Many of these studies have been focused on frequency flat channels for multiple-input multiple-output (MIMO) communication, e.g., see \cite{Risi}-\cite{Mo2}, \cite{Ivrlac}-\cite{Li2}. There have also been a few papers considering frequency selective channels for low resolution ADCs such as \cite{Studer}, which considers OFDM and concludes that $4$-$6$ bits ADCs are required to approach the performance of unquantized OFDM.  A complex nonlinear detector based on iterative turbo processing is proposed in \cite{Wang}, and is capable of detecting QPSK-data symbols in an OFDM waveform with $2$-$3$ bit ADCs, but it is ineffective when paired with a one-bit ADC. \cite{Mollen1}, \cite{Mollen2} demonstrated that one-bit ADCs in linear OFDM receivers for massive MIMO can give the same performance as one-bit ADCs for single carrier waveforms, provided there is an infinite number of channel taps.  Lastly, \cite{Mo-Schniter} studied the channel estimation for a few bit ADCs using the sparsity of the channel.

There has been a growing interest in harnessing the power of deep learning for applications in communication systems. Recently,  \cite{GeoffreyYeLi} presented a robust OFDM detection via deep learning against nonlinear impairments. For nonlinear channels, a recurrent neural network detector has been proposed in \cite{Nariman}. Furthermore, \cite{Oshea}-\cite{Felix} model the end-to-end communication system as an autoencoder for reliable detection without crafting complex modulation and coding schemes. We also use an autoencoder for OFDM detection. However, our work differs from those works because there is quantization before detection that creates a non-differentiable layer, and this hinders end-to-end training, which to our knowledge has not been considered previously. There are two recent papers that use tools from machine learning to handle low resolution quantization in MIMO systems \cite{Jeon1}-\cite{Nguyen}, but they do not consider OFDM and have a quite different receiver architecture. There are also some other MIMO studies that utilize learning based methods, e.g., \cite{Xiao-Chen}.

\subsection{Contributions}
\label{sec:cont}
This paper is composed of two main parts: (i) channel estimation and (ii) data detection. We propose different deep learning models for each part.  These specific models were selected for and adapted to the specifics of these different receiver tasks. 

\textbf{Channel estimation via a novel generative supervised deep learning model}.  We derive an expression to demonstrate that the channel would be estimated perfectly with one-bit ADCs if there was a very large number of pilots. Inspired by this expression, we produce a labeled data set, and train a deep neural network (DNN) accordingly with a limited number of training symbols for single antenna OFDM receivers. The key idea behind this model is to exploit the generalization property of neural networks to reduce the number of pilot symbols sent over the channel. In what follows, the trained DNN itself generates many output samples whose average gives the estimate of the channel taps in the frequency domain. This yields a generative learning model.  Using the formed data set, we first determine the number of sufficient training symbols for the proposed model, and then quantify its performance in terms of mean square error (MSE). Surprisingly, our proposed channel estimation model for one-bit quantized OFDM samples can give lower MSE than the least squares (LS) channel estimation with unquantized OFDM samples at average SNRs up to $14$ dB. 

\textbf{Data detection via an autoencoder that jointly learns a precoder and decoder.}  For data detection, we model the end-to-end OFDM communication system as a single autoencoder to jointly learn a precoder and decoder. However, this autoencoder cannot be trained in an end-to-end manner with the backpropagation algorithm due to the non-differentiable quantization layer. We tackle this problem by proposing a two step sequential training policy.  Accordingly, a decoder is first learned offline irrespective of the channel, then the precoder is learned online in conjunction with the trained decoder, taking into account the channel. The simulation results show the efficiency of the proposed method provided that the number of neurons in the hidden layers is moderately increased, which can be achieved by oversampling (still at one-bit resolution) in either the time or frequency domain. In particular, we can beat the theoretical bit error rate (BER) performance of uncoded unquantized QPSK-modulated data symbols in frequency selective Rayleigh fading at average SNRs up to $10$ dB when the dimension of the hidden layers before quantization is increased by a factor of $4$. 

\textit{Notation}: Matrices $\textbf{A}$ and vectors $\textbf{a}$ are designated as uppercase and lowercase boldface letters. $[\cdot]_{k,n}$ corresponds to the entry of a matrix in the $k^{th}$ row and $n^{th}$ column. Transpose and Hermitian operations are demonstrated by $(\cdot)^T$ and $(\cdot)^H$ respectively. The real and imaginary parts are $\Re(\cdot)$ and $\Im(\cdot)$. Trace of the matrix is referred as $tr[\cdot]$ and $\textbf{I}_N$ is $N \times N$ identity matrix.

\section{Channel Estimation with One-Bit ADCs}\label{Channel Estimation with One-Bit ADCs}
Reliable channel estimation with one-bit ADCs is challenging especially for OFDM, which mainly stems from the increased ICI. To tackle this problem, a novel generative supervised learning model is proposed. As a general rule, the efficiency of a supervised learning model depends on using an appropriate labeled data set, which is non-trivial. To determine a suitable labeled data, a theoretical analysis is done. Then, the proposed supervised learning model is grounded to this analysis to enable reliable channel estimation in OFDM receivers with one-bit ADCs.
\subsection{One-Bit OFDM Signal Analysis}
We assume that the channel experiences block fading.
\begin{figure} [!t] 
\centering 
\includegraphics [width=3.5in, angle = 0]{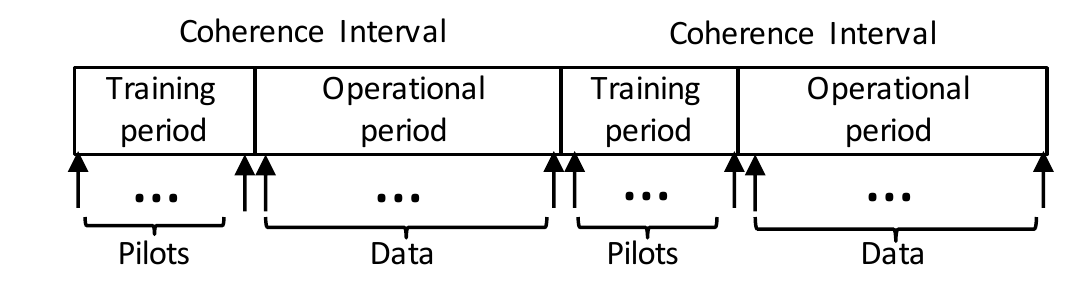}
\caption{Block fading channel model, in which data follows the pilots, and channel changes independently among blocks.}\label{fig:FadingModel}
\end{figure}
This channel is estimated through the pilot symbols $\textbf{s}_p$. These pilots are sent before data transmission starts at the beginning of each channel coherence time interval as demonstrated in Fig. \ref{fig:FadingModel}. The pilot symbols are multiplied by a normalized inverse discrete Fourier transform (IDFT) matrix, and transmitted over the dispersive channel after appending a cyclic prefix (CP). This can be expressed in complex matrix-vector form for $N$ subcarriers as
\begin{equation} \label{rec_chn_est}
\textbf{y}_p= \textbf{H}\textbf{F}^H\textbf{s}_p+\textbf{n}
\end{equation}
where $\textbf{s}_p=[s_{p_0} s_{p_1} \cdots s_{p_{N-1}}]^T$, $\textbf{F}$ is the normalized DFT matrix and so $\textbf{F}^H$ is the normalized IDFT matrix, $\textbf{H}$ is the $N\times N$ circulant channel matrix assuming the CP is removed at the receiver, and $\textbf{n}$ is zero-mean additive white Gaussian noise (AWGN) with variance $\sigma_n^2$. It is well-known that a circulant channel matrix has eigendecomposition 
\begin{equation} \label{circ}
\textbf{H}=\textbf{F}^H\boldsymbol{\Lambda}\textbf{F}
\end{equation}
where $\boldsymbol{\Lambda}$ is a diagonal matrix whose entries indicate the channel taps in the frequency domain, i.e.,
\begin{equation} \label{freqtaps}
H_i=\boldsymbol{\Lambda}_{i,i}
\end{equation}
for $i=0,\cdots, N-1$. 

One-bit quantization of (\ref{rec_chn_est}) with a pair of one-bit ADCs to quantize the real and imaginary part separately results in
\begin{equation} \label{quant}
\textbf{r}_p= \mathcal{Q}(\textbf{y}_p) = \frac{1}{\sqrt{2}}\text{sign}(\Re(\textbf{y}_p))+\frac{j}{\sqrt{2}}\text{sign}(\Im(\textbf{y}_p)).
\end{equation}
Outputs of an OFDM transmitter are time domain samples that can be well approximated by a Gaussian distribution \cite{JeffAndrewsBook}, and any nonlinear function of a Gaussian signal can be expressed in terms of the original Gaussian signal using Bussgang's theorem \cite{Bussgang}. Specifically, the quantization distortion can be defined as \cite{Mollen2}
\begin{equation} \label{quant_dist}
\textbf{d}_p = \textbf{r}_p - \textbf{A}\textbf{y}_p
\end{equation}
or equivalently
\begin{equation} \label{quant_sig}
\textbf{r}_p= \textbf{A}\textbf{y}_p+\textbf{d}_p
\end{equation}
wherein the matrix $\textbf{A}$ makes $\textbf{y}_p$ and $\textbf{d}_p$ uncorrelated to reduce the quantization noise \cite{Li2}, \cite{Mollen2}. That is,
\begin{equation} \label{exp_d_r}
E[\textbf{d}_p\textbf{y}_p^H] = E[\textbf{d}_p]E[\textbf{y}_p^H].
\end{equation}

\begin{lemma}\label{Lemma}
The quantization distortion and the pilots are uncorrelated, and
\begin{equation} \label{exp_d_s}
E[\textbf{d}_p\textbf{s}_p^H] = 0. 
\end{equation}
\end{lemma}
\begin{proof}
Taking the expected value of (\ref{rec_chn_est}) yields
\begin{equation}\label{exp_rec}
E[\textbf{y}_p] = 0
\end{equation}
because $E[\textbf{s}_p]=0$ and $E[\textbf{n}]=0$. Substituting (\ref{exp_rec}) in (\ref{exp_d_r}) gives
\begin{equation}\label{dyh}
E[\textbf{d}_p\textbf{y}_p^H] =0.
\end{equation}
Since the quantization distortion and channel noise are uncorrelated, using (\ref{rec_chn_est}) in (\ref{dyh}) trivially implies (\ref{exp_d_s}).
\end{proof}

\begin{theorem}\label{Theorem}
The diagonal matrix $\boldsymbol{\Lambda}$ can be obtained from the one-bit observations and pilots as 
\begin{equation}\label{Thm1}
E[\textbf{F}\mathcal{Q}(\textbf{y}_p)\textbf{s}_p^H]=\sqrt{ \frac{2}{\pi(\sigma_{chn}^2\sigma_{pilots}^2+\sigma_n^2)}}\sigma_{pilots}^2\ \boldsymbol{\Lambda}
\end{equation}
where 
\begin{equation}\label{rho}
\sigma_{pilots}^2=\frac{E[\textbf{s}_p^H\textbf{s}_p]}{N}
\end{equation}
and
\begin{equation}\label{rho}
\sigma_{chn}^2=\frac{tr[\boldsymbol{\Lambda}\boldsymbol{\Lambda}^H]}{N}.
\end{equation}
\end{theorem}
\begin{proof}
See Appendix.
\end{proof}

The channel can be estimated perfectly with a very large number of pilots that are sent for each channel coherence time interval with one-bit ADCs if instantaneous channel, pilots and noise powers are known. More precisely, if a large number of pilots are sent to estimate the channel, and each of these pilots is multiplied with the corresponding one-bit observation and the normalized DFT matrix, then taking the average of these terms can produce the $\boldsymbol{\Lambda}$ scaled by a scalar due to Theorem \ref{Theorem}. Since $\boldsymbol{\Lambda}$ is a diagonal matrix whose entries are the channel taps in the frequency domain, estimating $\boldsymbol{\Lambda}$ is equivalent to estimating the channel.

\subsection{Supervised Learning Model}
If there were many pilots in each channel coherence interval in addition to the instantaneous channel, pilots and noise power knowledge, the channel could be estimated perfectly. However, the number of pilots should be minimized to conserve bandwidth and power. It is also not practical to know the instantaneous channel power. Thus, a supervised channel learning model is proposed based on the idea of implementing (\ref{Thm1}) with a DNN. The underlying motivation for investigating a DNN architecture is associated with the generalization capability of DNNs \cite{Goodfellow}, which we will see greatly reduces the number of pilots that are necessary. 

We propose a DNN to estimate the channel as a regression task. The proposed DNN architecture is trained with special labeled data, in particular with the diagonals of the matrix $\textbf{F}\mathcal{Q}(\textbf{y}_p)\textbf{s}_p^H$. That is, the labeled data is produced via the pilot symbols, the corresponding one-bit quantized observations and the DFT matrix. This architecture is given in Fig. \ref{fig:chnLearn}, which is composed of an input layer, $2$ hidden layers and an output layer. Notice that a single hidden layer can give the same performance with two hidden layers if it has a sufficient number of neurons due to the universal function approximation theorem of neural networks \cite{Goodfellow}. However, this brings additional computational complexity, and hence having two hidden layers with reasonable number of neurons seems a good compromise. Fine-tuning our architecture is left to future work. The input layer takes the pilots $\textbf{s}_p$ and produces the corresponding output $\textbf{z}_p$ for $p=1, \cdots, N_t$ where $N_t$ is the total number of pilots transmitted over the channel for one coherence interval. $\textbf{z}_p$ can be written in terms of the trainable weights or network parameters (in matrix notation) and activation functions as
\begin{equation} 
\textbf{z}_{p}=\sigma_3(\boldsymbol{\Theta}_{3}\sigma_2(\boldsymbol{\Theta}_{2}\sigma_1(\boldsymbol{\Theta}_{1} \textbf{s}_{p}))).
\end{equation}
The parameters are optimized according to the following cost function
\begin{equation} \label{cost_func}
J=\min_{\boldsymbol{\Theta}_{1}, \boldsymbol{\Theta}_{2}, \boldsymbol{\Theta}_{3}}\left|\left|{\textbf{z}_p}-\text{diag}(\textbf{F}\mathcal{Q}(\textbf{y}_p)\textbf{s}_p^H)\right|\right|^2
\end{equation}
which are solved with gradient descent via the backpropagation algorithm.

\begin{figure} [!t] 
\centering 
\includegraphics [width=3.5in]{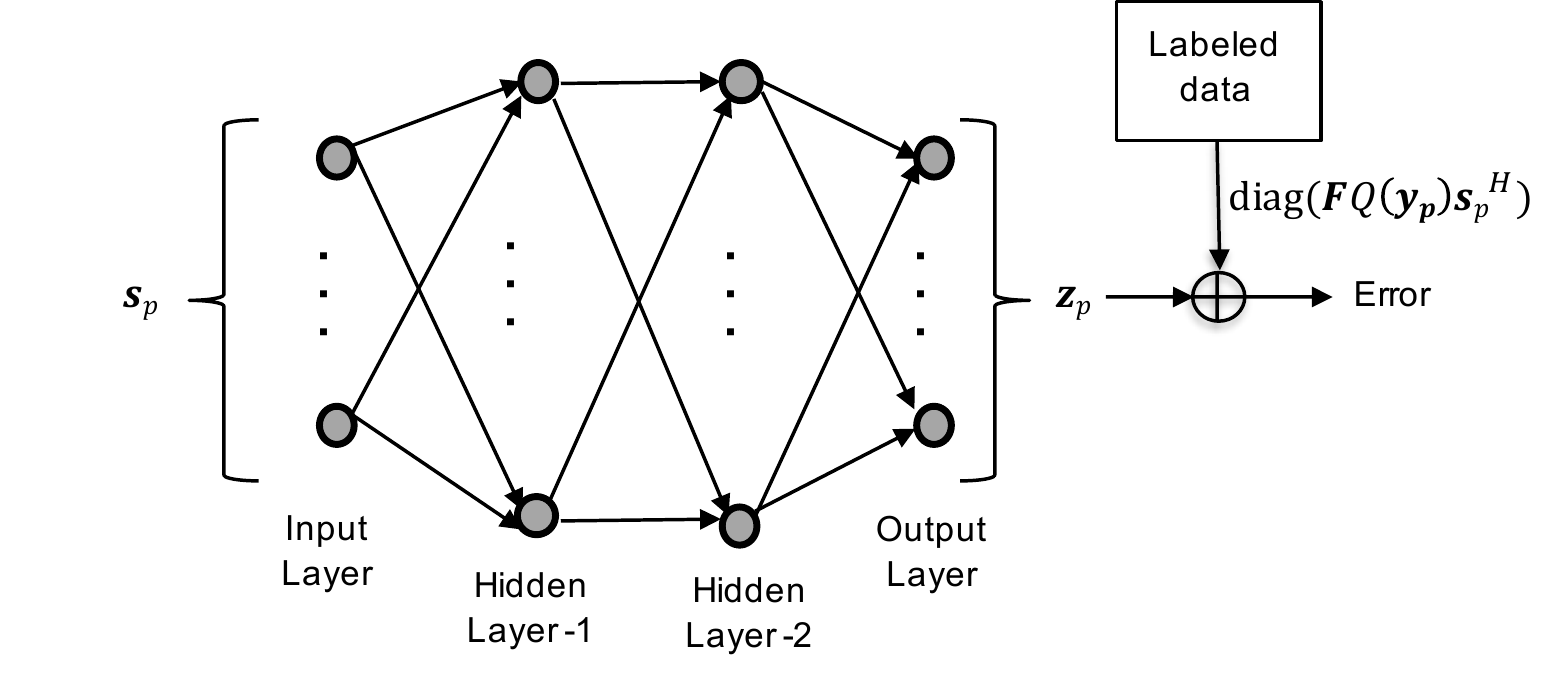}
\caption{The inputs, outputs and labeled data for the proposed DNN. }\label{fig:chnLearn}
\end{figure}

The layers, their types, sizes, activation functions and weights are summarized in Table \ref{tab:SLM}. Since state-of-the-art software libraries that implement neural networks do not support complex operations, the real and imaginary part of the complex vectors are concatenated to obtain a $2N\times1$ real vector. Without loss of generality, the dimension of the hidden layers is taken to be twice that of the input and output layer, giving $32N^2$ trainable parameters, which increases quadratically with the number of subcarriers. Rectified linear unit (ReLU) is used in the hidden layers as an activation function for fast convergence, and a linear activation function is utilized at the output layer, because this is a regression task. The weights between the two layers are specified by matrices.
\begin{table} [!t] 
\renewcommand{\arraystretch}{1.3}
\caption{The proposed DNN architecture for channel estimation with  one-bit ADC}
\label{tab:SLM}
\centering
\begin{tabular}{c|c|c|c|c}
    \hline
       Layer & Type & Size & Activation & Weights\\
    \hline
    \hline
Input Layer & Pilot Symbols & 2N & - & $-$\\
    \hline
Hidden Layer-1 & Fully Connected & 4N & ReLU & $\boldsymbol{\Theta}_{1}$\\
    \hline
 Hidden Layer-2 & Fully Connected & 4N & ReLU & $\boldsymbol{\Theta}_{2}$\\
    \hline
 Output & Fully Connected & 2N & Linear & $\boldsymbol{\Theta}_{3}$\\
		\hline
\end{tabular}
\end{table}

The DNN is trained to minimize the MSE between the outputs and the labeled data as given in (\ref{cost_func}). This implies that the learned probability distribution of the output can approximate the probability distribution of $\text{diag}(\textbf{F}\mathcal{Q}(\textbf{y}_p)\textbf{s}_p^H)$. Once the model is trained, we generate as many output samples as needed from the learned distribution in response to random inputs within the same channel coherence interval, and take their average to estimate the channel in accordance with (\ref{Thm1}). The generated output samples for the random inputs do not cost anything other than some extra processing, because these inputs are not coming from the channel; rather they are generated randomly in the receiver. Note that many different types of generative model applications emerge after the seminal paper of \cite{GANs} proposed to train a generative model in the framework of a generative adversarial network (GAN). To be more precise, our trained DNN generates some output samples $\textbf{z}_i$ in response to the random inputs $\textbf{s}_i$. In what follows, the channel taps in the frequency domain are estimated as
\begin{equation} \label{est_chn_taps}
\hat{H} = \frac{1}{M}\sum_{i=0}^{M-1}\textbf{z}_i
\end{equation}
where $\hat{H}=[\hat{H}_0 \cdots \hat{H}_{N-1}]$. Note that $M$ is the total number of arbitrarily generated output samples. There is no constraint to limit $M$ except the processing complexity, i.e., the $\textbf{z}_i$ does not consume any bandwidth. Note that at each time the channel changes, the model must be retrained with $N_t$ pilots, and $M$ randomly generated samples after training the DNN with the pilots.

The overall computational complexity of the proposed channel estimation model is composed of training the DNN model and generating random samples from the trained DNN. The former leads to the complexity of $\mathcal{O}(W^2)$ where $W=32N^2$ is the total number of adaptive parameters in the DNN, which stems from the backpropagation algorithm. The latter phase has relatively less complexity, in particular its complexity comes from matrix-vector multiplication. Hence, the proposed learning model for channel estimation has a complexity of $\mathcal{O}(W^2)$.

\section{Data Detection with One-Bit ADCs}
\label{Data Detection with One-Bit ADCs}
Reliably detecting the OFDM symbols with one-bit ADCs is extremely difficult even if channel is estimated and equalized perfectly because of the resulting severe ICI. The ICI results because quantization in the time domain disrupts the orthogonality between the subcarriers in the frequency domain. For example, consider QPSK modulated OFDM symbols transmitted over a $10$-tap frequency selective channel at $20$ dB SNR. This yields the constellation diagram given in Fig. \ref{fig:scatter_unquant_cnst} and Fig. \ref{fig:scatter_one_bit_cnst} for the unquantized and one-bit quantized received samples assuming that the channel is perfectly estimated and equalized for both cases. It does not seem possible to reliably detect these QPSK symbols with one-bit ADCs. 

\begin{figure}[!h]
\centering
\subfigure[Unquantized OFDM]{
\label{fig:scatter_unquant_cnst}
\includegraphics[width=1.5in]{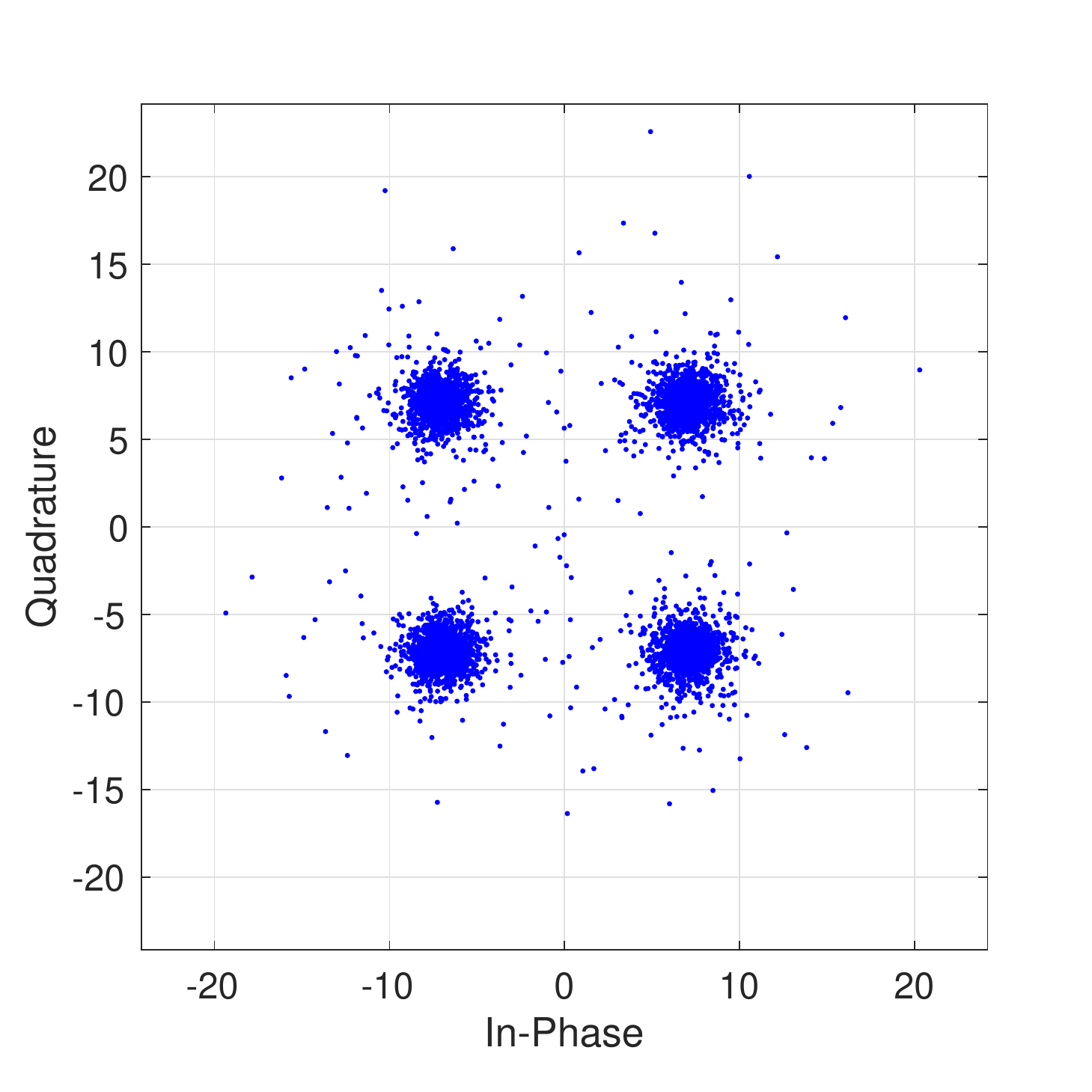}}
\qquad
\subfigure[OFDM with one-bit ADCs]{
\label{fig:scatter_one_bit_cnst}
\includegraphics[width=1.5in]{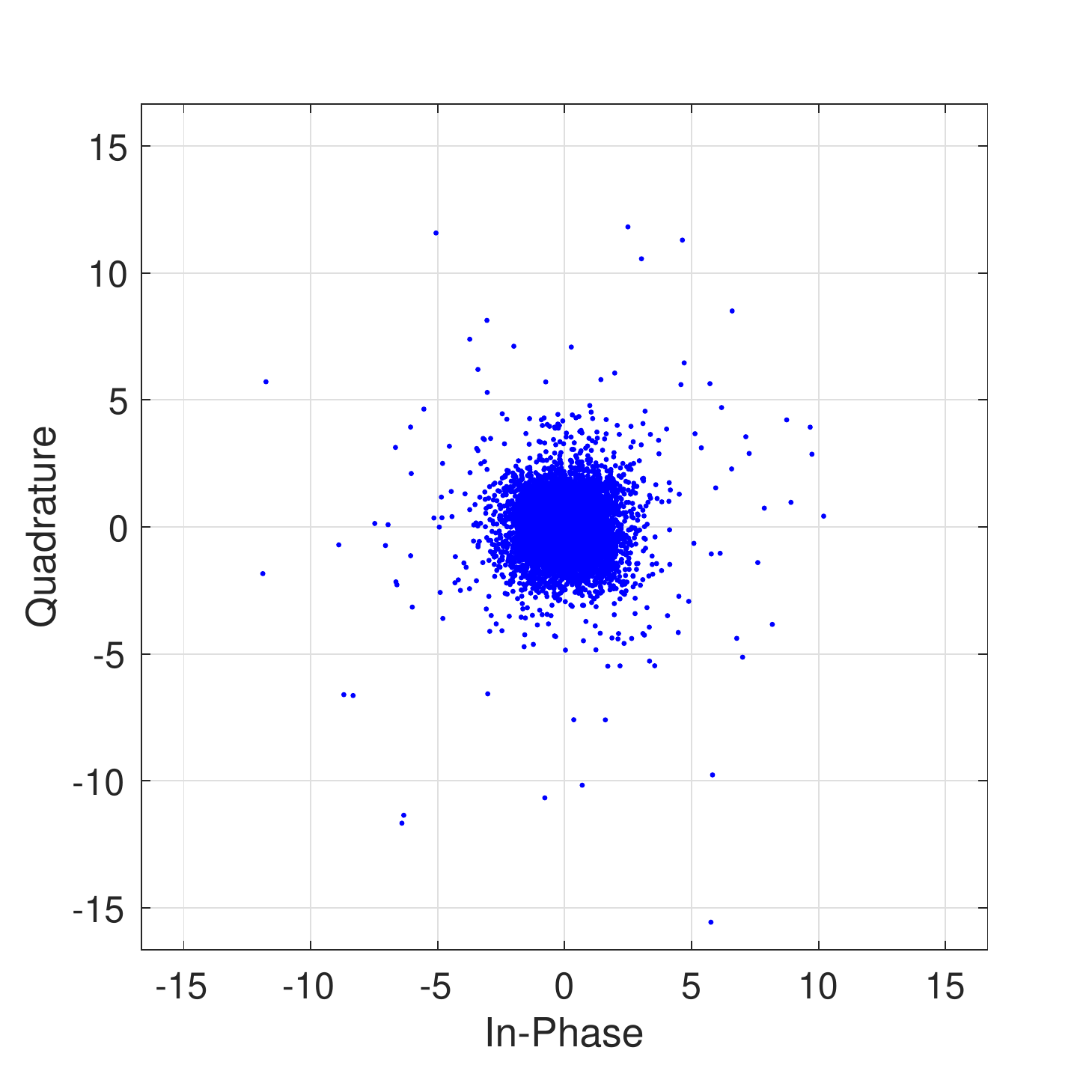}}
\caption{Constellation diagram of the QPSK modulated OFDM symbols received at $20$ dB SNR for the (a) ideal unquantized case (b) one-bit quantization applied separately for the in-phase and quadrature terms.}
\end{figure}

To have a satisfactory error rate for the detection of OFDM symbols with one-bit ADCs, we propose to jointly learn a precoder and decoder. This can be done by adapting an autoencoder, which is a powerful unsupervised deep learning tool, to the OFDM system. More precisely, the end-to-end OFDM communication system is treated as a single autoencoder to jointly learn a precoder and decoder. The main challenge related with this approach surfaces in training. Specifically, one-bit ADCs lead to a non-differentiable layer in the autoencoder, which hinders the training of the parameters. This issue is handled via a novel two-step sequential training policy. The practical challenges and our implementation suggestions for the aforementioned model are given at the end of this section.

\subsection{Autoencoder Based OFDM}\label{Autoencoder}
An autoencoder aims to copy its inputs to the outputs by decreasing the signal dimension in the hidden layers so as to enforce a sparse representation of the input \cite{Goodfellow}. By this it is meant that autoencoders can reconstruct the output from a low-dimensional representation of the input\footnote{Here, the low dimension refers to the low resolution data not the number of neurons in the hidden layers.} at some hidden layer by learning an encoder and decoder. This is a good match for our problem, in which the transmitted OFDM symbols are detected using the one-bit quantized observations with the help of a precoder and decoder. Here the analogy is that the OFDM symbols correspond to the inputs, the one-bit quantized data is a hidden layer, and the outputs represent the detected symbols.

\begin{figure*} [!t] 
\centering 
\includegraphics [width=6in]{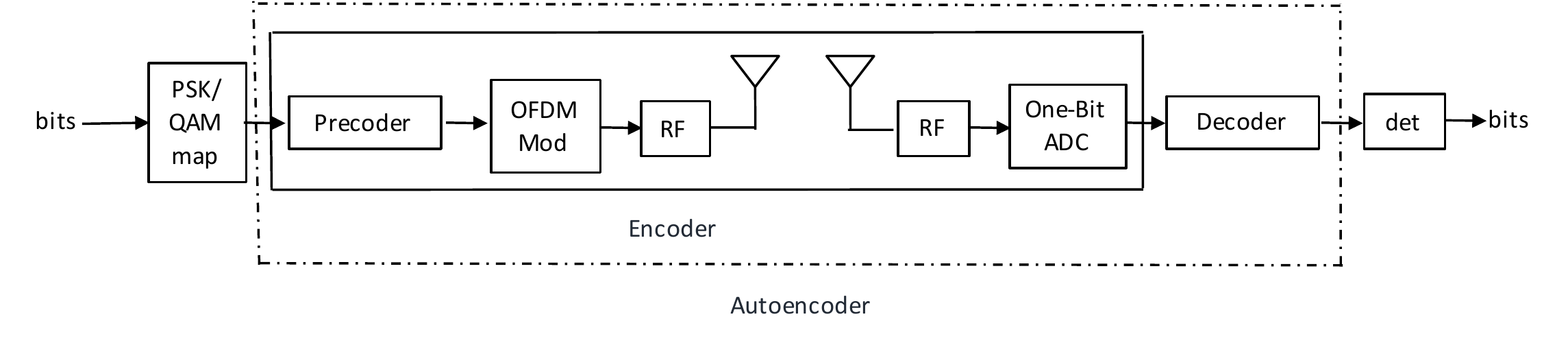}
\caption{AE-OFDM: Adapting an autoencoder for OFDM systems with one-bit ADCs.}\label{fig:AEmodel}
\end{figure*}
To make use of an autoencoder for one-bit OFDM detection, the main building blocks, which are the encoder and decoder, have to be adapted. Therefore, the learned precoder \textbf{P}, OFDM modulator $\textbf{F}^H$ (which is realized as an IDFT), channel \textbf{H}, noise and quantizer can be seen collectively as an encoder. The decoder corresponds to the post-processing after quantization at the receiver. This model is shown in Fig. \ref{fig:AEmodel}, and termed as AE-OFDM, which is consistent with state-of-the-art OFDM transceivers except the precoder and decoder are now implemented as artificial neural networks. 

In AE-OFDM, the modulated symbols at the $N$ subcarriers, i.e., $s_k$ for $k=0,1,\cdots,N-1$ are multiplied with a linear precoder matrix in the frequency domain, which will be learned through training. This leads to
\begin{equation}
{\textbf{x}}=\textbf{P}\textbf{s}
\end{equation}
where $\textbf{P} \in \mathbb{C}^{N\times N}$ is the frequency domain precoder matrix, and $\textbf{s}=[s_0 s_1 \cdots s_{N-1}]^T$. Crucially, the pilot symbols are not multiplied by a precoder matrix in channel estimation, since the precoder is designed according to the channel, i.e., after channel estimation. In what follows, an IDFT is applied to the precoded symbols, and transmitted over a dispersive channel that has $L$ time domain taps such that $L<N$. This results in
\begin{equation} \label{RecSigAE}
\textbf{y}= \textbf{H}\textbf{F}^H\textbf{x}+\textbf{n}
\end{equation}
which is similar to (\ref{rec_chn_est}) except the pilot symbols $\textbf{s}_p$ are replaced with $\textbf{x}$. How the channel taps can be estimated via deep learning was given in (\ref{est_chn_taps}). 

One-bit quantization of (\ref{RecSigAE}) with a pair of ADCs for the in-phase and quadrature components provides the input to the decoder 
\begin{equation} 
\textbf{r}=\mathcal{Q}(\textbf{y})
\end{equation}
such that $\mathcal{Q}(\cdot)$ is applied element-wise. The decoder \textbf{D} is a multi-layer neural network whose aim is to reconstruct \textbf{s} from \textbf{r}. Specifically,
\begin{equation} 
\textbf{s}' =\sigma_Z(\boldsymbol{W}_{Z}\cdots\sigma_2(\boldsymbol{W}_{2}\sigma_1(\boldsymbol{W}_{1} \textbf{r})))
\end{equation}
where $Z$ is the number of layers and $\sigma_z$ is the activation function for layer $z$ applied element-wise for vectors. The dimension of the parameter matrices is 
\begin{align} \label{rewardFunc}
    \dim(\boldsymbol{W}_{z}) = 
    \begin{cases}
       l_z \times \dim(\textbf{y}), &  z=1\\
       l_z \times l_{z-1}, &  z=2,\cdots, Z-1 \\
       N \times l_{Z-1} &  z=Z
    \end{cases} 
\end{align}

In summary, the end-to-end AE-OFDM architecture from the transmitter to the receiver can be divided into logical blocks as depicted in Fig. \ref{fig:genDL}. Here, the modulated symbols are treated as an input layer and the detected symbols constitute the output layer. AE-OFDM eliminates the need at the receiver for an explicit DFT and equalization, because they are implicitly learned. The next step is to learn the neural network precoder $\textbf{P}$ and decoder $\textbf{D}$ by properly training the model. 
\begin{figure} [!t] 
\centering 
\includegraphics [width=3.5in]{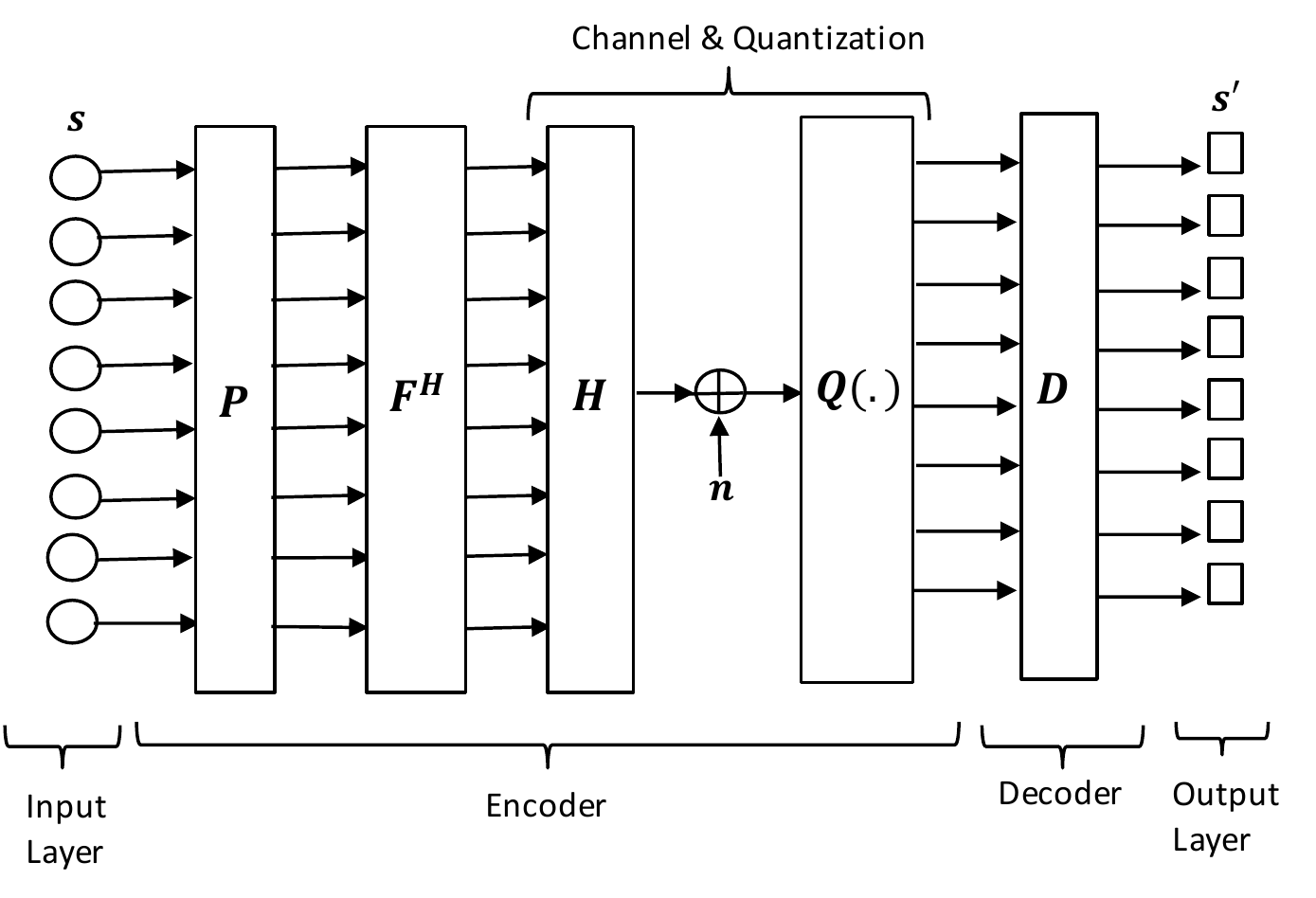}
\caption{The block diagram of the AE-OFDM architecture with one-bit quantization.}\label{fig:genDL}
\end{figure}

\subsection{Training}\label{Training}
Autoencoders are trained to minimize the reconstruction loss or the sum-of-squares error function between the input and output layer, which corresponds to 
\begin{equation} \label{recloss}
e= ||\textbf{s}-\textbf{s}'||^2
\end{equation}
where $\textbf{s}'=[s'_0 s'_1 \cdots s'_{N-1}]^T$. The parameters of the neural layers in the precoder and decoder are trained according to this error function as
\begin{equation} 
W_{k,l}^{(n+1)}=W_{k,l}^{(n)}-\mu\frac{\partial e}{\partial W_{k,l}^{(n)}},
\end{equation}
where $W_{k,l}^{(n)}$ indicates the $l^{th}$ neuron at the $k^{th}$ layer in the $n^{th}$ iteration and $\mu$ is the learning rate. The gradient of the error is evaluated using a local message passing scheme among layers known as backpropagation. However, the quantization layers or $\mathcal{Q}(\cdot)$ stymies the backpropagation, because its derivative is $0$ everywhere except that the point at $0$ that is not even differentiable. Thus, any neural layer before $\mathcal{Q}(\cdot)$, which corresponds to the precoder, cannot be trained. Hence, a novel training policy is needed for the AE-OFDM model.

In this paper, a two-step sequential learning model is proposed to train the AE-OFDM instead of end-to-end training. In the first step, the decoder is trained without explicitly considering the channel and OFDM modulator. In the second step, the precoder is learned to be compatible with the trained decoder taking into account the channel and OFDM modulator. An apparent advantage of this training policy lies in the fact that the decoder can be trained offline, which brings significant complexity savings. On the other hand, the precoder has to be learned online at each time the channel changes. This can be done with a reasonable pre-determined number of training samples with a small size neural network following the channel estimation. In particular, both the decoder and precoder are trained with $5000$ samples for an OFDM system that has $64$ subcarriers. Note that this does not mean that $5000$ pilots symbols are sent over the channel, as will be explained in the next section.

The overall end-to-end model for the two-step sequential training policy including all layers from $l_1$ to $l_8$ is given in Fig. \ref{fig:E2E}. 
\begin{figure} [!t] 
\centering 
\includegraphics [width=3.5in, angle = 0]{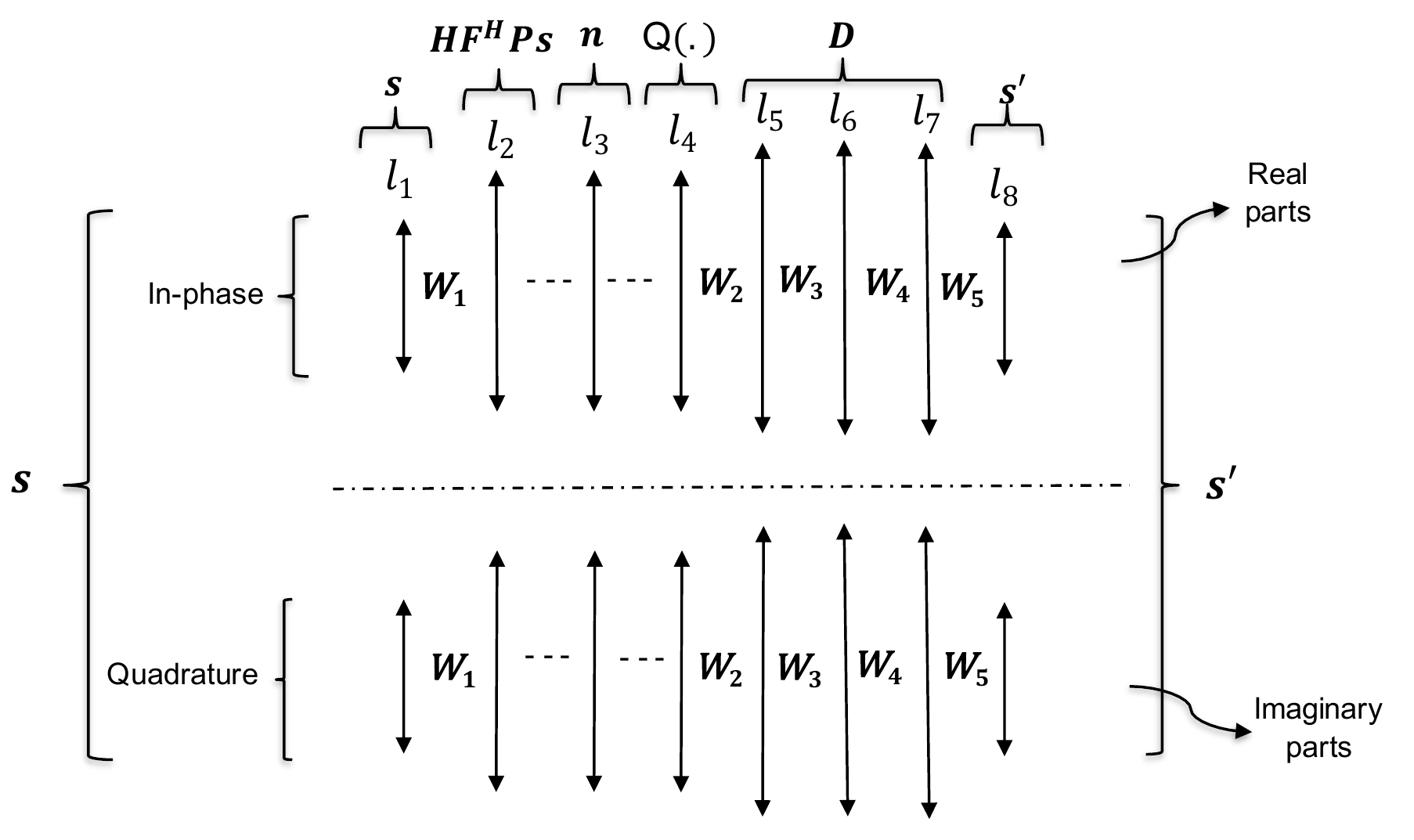}
\caption{The end-to-end layered architecture of the AE-OFDM.}\label{fig:E2E}
\end{figure}
Notice that each layer is composed of two parts. More precisely, $l_i$ for $i=1,\cdots, 8$ can be considered as a single vector, wherein the first half corresponds to the in-phase components of the symbols or the real part of the complex baseband signal denoted as $l_{iI}$, and the second half represents the quadrature or imaginary part represented by $l_{iQ}$. Accordingly, the symbols are first multiplied with the learned precoder matrix. The precoded symbols are normalized to ensure the average transmission power constraint, and then multiplied with the IDFT and channel matrix, respectively, which forms the $l_2$ layer. To obtain $l_3$, AWGN is added both for the real and imaginary parts, and the resultant samples are quantized in $l_4$. Lastly, the decoder processes the data via $l_5,l_6,l_7$, and outputs are obtained at $l_8$. The decoded symbols are mapped to the closest constellation point according to the minimum Euclidean distance criterion. 

For the decoder, the complex baseband signal can be easily divided into real and imaginary parts, each of which is processed separately with the same set of parameters. That is, there is parameter sharing, which is one of the key concepts behind the success of deep learning. The rationale behind parameter sharing is to decrease the complexity. To be specific, $\boldsymbol{W}_1, \boldsymbol{W}_2, \boldsymbol{W}_3, \boldsymbol{W}_4, \boldsymbol{W}_5$ demonstrate the shared parameters. Although the received complex OFDM baseband signal can be trivially broken into real and imaginary parts for decoder, it is not straightforward to divide the signal at the transmitter. This is associated with the OFDM modulation that mixes the in-phase and quadrature parts of the modulated symbols via the IDFT. This challenge is inherently handled while training the precoder, in which we implement a simple supervised learning model by using the $l_2$ layer of the decoder as a labeled data set for the input $l_1$. It is worth emphasizing that although the overall autoencoder architecture is an unsupervised learning model, supervised learning is used within this autoencoder so as to train the encoder part in case of one-bit quantization. As noted earlier, the one-bit quantizers prevent end-to-end training.

A supervised learning model is presented in Fig. \ref{fig:preLearn} to train the precoder associated with the trained decoder.  
\begin{figure} [!t] 
\centering 
\includegraphics [width=3.5in, angle = 0]{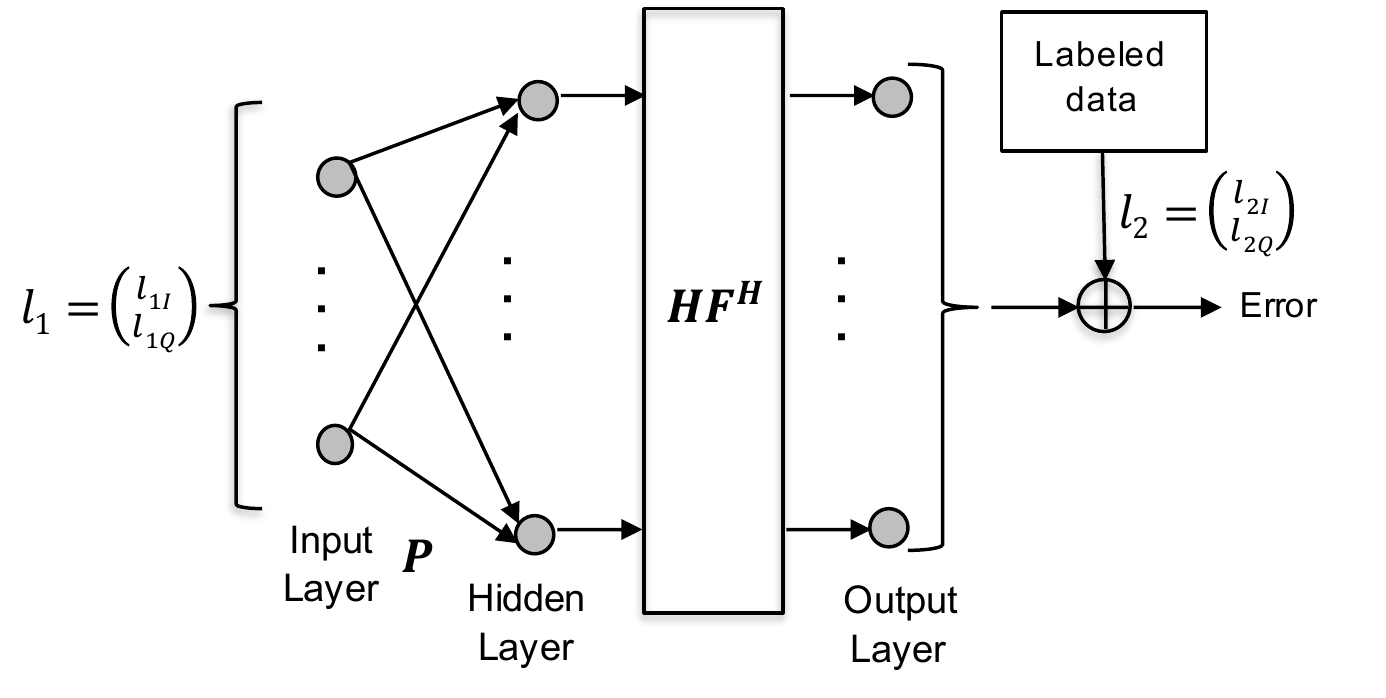}
\caption{A supervised learning model for the precoder that uses the $l_1$ and $l_2$ layers of the decoder in the training phase.}\label{fig:preLearn}
\end{figure}
During the training of decoder, the real and imaginary parts of $l_1$ and $l_2$ are concatenated to obtain a real vector with dimension $2N$. Then, these values of $l_1$ and $l_2$ layers are stored to create a data set to train the precoder such that $l_1$ constitutes the input data, and $l_2$ is used for labeled data. In this model, the inputs are processed with a neural layer, which corresponds to the precoder. Then, the precoded symbols are transformed to another vector by multiplying it with $\textbf{H}\textbf{F}^H$. Since the precoder is trained after estimating the channel, $\textbf{H}$ is already known. The primary aim of this model is to learn the output samples with respect to the labeled data set through the learned precoder. Theoretically, the labeled data set can be very well approximated with the outputs, because in this case there are no factors that limit the learning such as noise, data impediments, or dimension reduction. 

The further details of the AE-OFDM architecture including the layer types, sizes, activation functions and trainable weights of layers are illustrated in Table \ref{tab:AE-OFDM} considering the in-phase and quadrature parts separately. The layers before quantization can have a higher number of neurons than the input to make the learning more efficient, i.e., their size is $GN$ such that $G\geq 1$. This can be achieved with oversampling. Similarly, the decoder layers have a high dimension as $KN$, in which $K$ is taken $20$ without any loss of generality. Note that our empirical observations demonstrate that the value of $G$ affects the performance much more than $K$. Hence, the results are obtained for different values of $G=\{1, 2, 4\}$. At the output, a linear activation function is used, and thus a continuous valued vector with $N$ terms is obtained. Each term of this vector is individually mapped to one of the constellation points according to the minimum Euclidean distance criterion. This greatly reduces the dimension of the output when compared to using softmax activation function at the output in conjunction with a one-hot encoding, since this requires a $2^N$ dimensional output vector.

\begin{table*} [!t] 
\renewcommand{\arraystretch}{1.3}
\caption{AE-OFDM Model and Layers}
\label{tab:AE-OFDM}
\centering
\begin{tabular}{c|c|c|c|c|c}
    \hline
     & Layer & Layer Type & Size & Activation & Weights \\
    \hline
    \hline
    $l_{1I}$ - $l_{1Q}$ &Input & Input Symbols & $N$ & - &-\\
    \hline
         \multirow{3}{*}{$l_{2I}$ - $l_{2Q}$} & Precoder &  &  &  &\\
    
	 & IDFT & Fully Connected & $GN$ & Linear & $\boldsymbol{W}_1$ \\
    
 	 & Channel &  &  &  &  \\
    \hline
     $l_{3I}$ - $l_{3Q}$ &Noise & Vector & $GN$ & - \\
    \hline
     $l_{4I}$ - $l_{4Q}$ &Quant & Sign Function & $GN$ & - \\
    \hline
     $l_{5I}$ - $l_{5Q}$   & \multirow{3}{*}{Decoder}  & Fully Connected & $KN$ & ReLU &$\boldsymbol{W}_2$\\
    
	$l_{6I}$ - $l_{6Q}$ &  & Fully Connected & $KN$ & ReLU &$\boldsymbol{W}_3$ \\
    
	$l_{7I}$ - $l_{7Q}$ &  & Fully Connected & $KN$ & ReLU&$\boldsymbol{W}_4$ \\
    \hline
		$l_{8I}$ - $l_{8Q}$ &  Output & Fully Connected & $N$ & Linear &$\boldsymbol{W}_5$\\
		\hline
\end{tabular}
\end{table*}
 
The computational complexity of the proposed learning model for data detection is $\mathcal{O}(W_d^2)$ + $\mathcal{O}(W_e^2)$, where $W_d$ is the number of adaptive parameters in the decoder and $W_e$ is the number of adaptive parameters in the encoder. Since we employ a two-step sequential training policy, the decoder parameters are trained in the first step and the parameters of the encoder are trained in the following step according to the trained decoder. This is illustrated in Fig. \ref{fig:preLearn}.
 
\subsection{Practical Challenges} \label{Challenges}
The channel varies according to the block fading model so the precoder parameters have to be retrained each time the channel changes. This can bring excessive training symbol overhead. As a worst case assumption, each OFDM symbol could be required to occur at least once in the training phase.  This would require at least $2^N$ pilots, which in our case would be greater than $10^{19}$.  This shows the generalization capability of the DNN that will be trained with just $5000$ symbols in this paper for $N=64$. However, this number of pilots is still impractical in terms of bandwidth efficiency. To address this problem, one solution can be to train the precoder at the receiver after estimating the channel. Then, the learned precoder can be notified to the transmitter before data transmission begins. This brings the flexibility of training the model with as many samples as needed without decreasing the bandwidth efficiency due to additional pilots. The main drawback of this training model can be some extra processing at the receiver. However, this complexity can be handled using stochastic computing-based hardware implementations \cite{Kim}. Another solution can be to train the precoder parameters at the transmitter after the receiver sends the channel information to the transmitter. This can be especially useful in downlink communication. With this approach, there is no need to do training in the receiver, since the precoder is trained in the transmitter and the decoder parameters are trained offline.

\subsection{Implementation} \label{Implementation}
The layered model gives an abstract view of AE-OFDM, which means that it can be implemented in many different ways in practical transceivers, in particular, depending on how the input dimension is increased when $G>1$. Adding redundant subcarriers, employing multiple antennas, oversampling in time and/or in frequency domain are methods to increase the input dimension. In this paper, our focus is on oversampling methods, wherein $G$ is treated as the oversampling factor so that AE-OFDM can be realized by either time domain oversampling or frequency domain oversampling, which we now discuss in turn.

\subsubsection{Time Domain Oversampling} 
The discrete-time received signal can be written as  
\begin{equation}\label{signaltime}
y_n=\sum_{l=0}^{L-1}h_lx_{n-l}+n_n
\end{equation}
where $h_l$ is the channel taps in the time domain, $n_n$ is the complex Gaussian noise as $CN(0,\sigma_n^2)$, and
\begin{equation}\label{precodedSymbols}
x_n = \frac{1}{\sqrt{N}}\sum_{k=0}^{N-1}X_ke^{j2\pi kn/N}
\end{equation}
in which $X_k$ is the precoded symbol in the frequency domain.

The received continuous-time complex signal can be expressed analogous to (\ref{signaltime}) as 
\begin{equation}
y(t) = \frac{1}{\sqrt{N}}\sum_{k=0}^{N-1}H_kX_ke^{j2\pi kt/T}+n(t)
\end{equation}
where $T$ is the OFDM symbol period, and 
\begin{equation}
H_k=\sum_{l=0}^{L-1}h_le^{-j2\pi kl/N}.
\end{equation}
This signal is sampled at time instances $t = nT_s+ gT_s/G$ where $T_s=T/N$ and $g=0,1,\cdots,G-1$, which produces
\begin{equation}\label{yng}
y_{n_g} = y(nT_s+ gT_s/G). 
\end{equation}
Generalizing (\ref{yng}) to the matrix-form leads to
\begin{equation}
\textbf{y} = \textbf{H}_{\rm tos}\textbf{F}^H\textbf{P}_{\rm tos}\textbf{s}+\textbf{n}
\end{equation}
where $\textbf{y} = [\textbf{y}_0  \textbf{y}_1 \cdots \textbf{y}_{G-1}]^T$ such that $\textbf{y}_g = [y_{g_0}  y_{g_1} \cdots y_{g_{N-1}}]^T$, and
\[
\textbf{H}_{\rm tos} = 
\left[ \begin{array}{c}
   \textbf{F}^H\textbf{E}_0\textbf{F}\textbf{H}_{\rm srs}  \\
   \vdots  \\
   \textbf{F}^H\textbf{E}_{G-1}\textbf{F}\textbf{H}_{\rm srs} \\
\end{array}\right]
\] 
where
\begin{equation}
\textbf{E}_g = \text{diag}(1, e^{j2\pi g/GN}, e^{j4\pi g/GN}, \cdots, e^{j2\pi (N-1)g/GN}).
\end{equation}
Hence, the oversampled channel matrix $\textbf{H}_{\rm tos}$ can be written in terms of the symbol rate sampled channel matrix $\textbf{H}_{\rm srs}$, which becomes $GN\times N$, where $G$ shows the time domain oversampling factor. In this case, the precoder matrix $\textbf{P}_{\rm tos}$ remains complex $N\times N$ matrix as $\textbf{F}^H$. Note that $\textbf{P}_{\rm tos}$ can be learned according to $\textbf{H}_{\rm tos}$, and this results in   
\begin{equation}
s_{2}=\textbf{H}_{\rm tos}\textbf{F}^H\textbf{P}_{\rm tos}s_{1} 
\end{equation}
where $s_{1}=l_{1I}+jl_{1Q}$ and $s_{2}=l_{2I}+jl_{2Q}$. In what follows, the real and imaginary parts of $s_2$ are concatenated to obtain the real vector $l_2$ that is used by the decoder to detect the transmitted symbols.

\subsubsection{Frequency Domain Oversampling} 
Zeros are padded at the transmitter before IDFT to realize frequency domain oversampling. This obviously increases the block size of the IDFT by a factor $G$. In this case, the precoder matrix is found according to the frequency domain oversampled channel, and this produces
\begin{equation}
s_{2}=\textbf{H}_{\rm fos}\textbf{F}_{\rm fos}^H\Gamma\textbf{P}_{\rm fos}s_{1} 
\end{equation}
where $\textbf{P}_{\rm fos}$ is a $N\times N$ matrix, and
\[
\Gamma= 
\left[ \begin{array}{c}
   \textbf{I}_{N\times N}  \\
   \textbf{0}_{(G-1)N\times N}   \\
\end{array}\right].
\] 
Further, $\textbf{H}_{\rm fos}$ and $\textbf{F}_{\rm fos}$ are $GN\times GN$ matrices. 

AE-OFDM can also be implemented as a combination of time and frequency domain oversampling, and  the precoder matrix can be found accordingly. In this case
\begin{equation}
G=G_tG_f 
\end{equation}
where $G_t$ and $G_f$ denote the oversampling factor in time and frequency domain, and the matrix representations can be obtained trivially via the derived expressions. To summarize, AE-OFDM can be implemented in many different ways, and this choice depends on the requirements of communication schemes. For example, if AE-OFDM operates in the sub $6$-GHz with moderate bandwidth, time domain oversampling can be done without increasing the power consumption much due to the increased sampling rate \cite{Walden}. On the other hand, frequency domain oversampling can be preferred for mmWave transmissions that provides large bandwidth, in which the high sampling rate can be too costly regarding the power consumption at the expense of implementing longer IDFT and DFT. 

\section{Simulations} \label{Simulations}
The proposed generative supervised deep learning model for channel estimation, and unsupervised autoencoder model for data detection are evaluated using tensors to make use of TensorFlow framework while implementing neural layers. Note that a tensor can be viewed as $n$-dimensional arrays involving matrices or vectors, in which TensorFlow can run computations over them. The efficiency of the proposed models are assessed by generating a synthetic data for the transmitted symbols, wireless channel and noise. It is assumed that transmitted symbols are QPSK modulated, wireless channel taps are complex Gaussian, and they have uniform power delay profile. Noise samples are additive white Gaussian random variables. There are $64$ subcarriers in one OFDM block, i.e., $N=64$. This is consistent with IEEE 802.11a/g/n/ac, and could also be reasonable for the LTE downlink, since a given UE is often allocated one or two resource block groups, which are each $36$ subcarriers (for a 10 MHz bandwidth). For the models, the performance metric for channel estimation is MSE, and it is BER for data detection. 

\subsection{Channel Estimation}
The DNN model for the channel estimation given in Table \ref{tab:SLM} is trained with $3$ different number of training symbols or pilots transmitted over the channel as $N_t=\{10, 20, 25\}$ to determine the sufficient number of training symbols. In training, gradient descent is used with an adaptive learning rate, wherein gradients are found with backpropagation algorithm, and Adam optimizer is employed to have an adaptive learning rate whose initial learning rate is $0.01$ \cite{Adam}. Once the DNN is trained according to this setting, $M=10,000$ randomly generated input samples are input to the DNN, and their corresponding $10,000$ outputs are averaged to estimate the channel taps in the frequency domain. Note that our empirical results show that $M$ can be much less than $10,000$ provided there are sufficient number of pilots. This simulation is repeated for $100$ different channel realizations. Then, its performance is compared with the state-of-the art LS channel estimation for unquantized OFDM samples and one-bit quantized OFDM samples. The proposed generative deep learning-based model is also compared with the optimum maximum likelihood channel estimation for unquantized samples. Comparing the performance with an unquantized maximum likelihood channel estimation shows how efficiently the proposed model can cope with the detrimental effects of quantization. In the ideal case, a generative deep learning model can perfectly estimate the unquantized samples from the quantized samples, which can achieve the maximum likelihood channel estimation performance.

\begin{figure*}[!t]
\centering
\subfigure[For $3$ complex Gaussian channel taps]{
\label{fig:threetaps}
\includegraphics[width=3in]{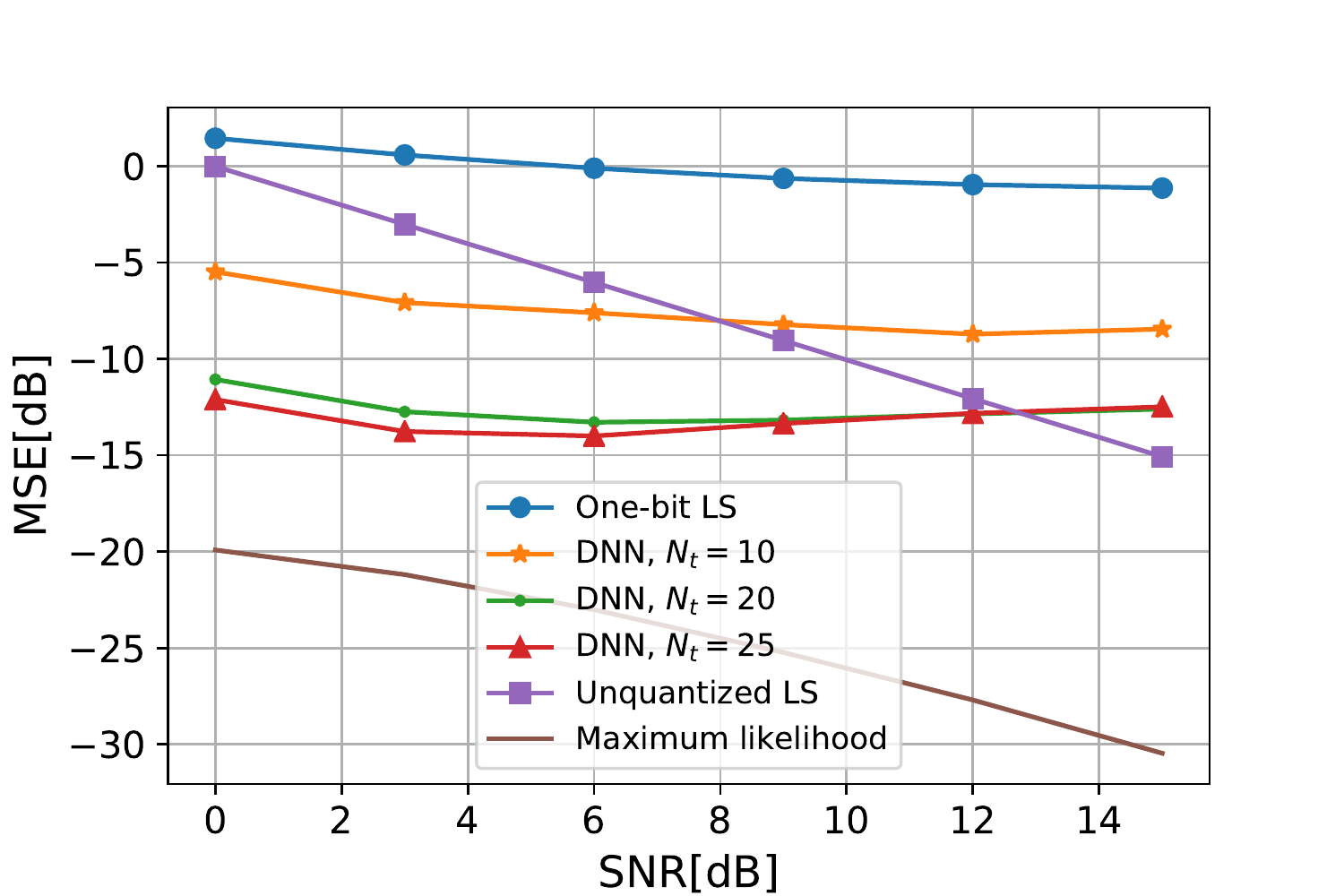}}
\qquad
\subfigure[For $10$ complex Gaussian channel taps]{
\label{fig:tentaps}
\includegraphics[width=3in]{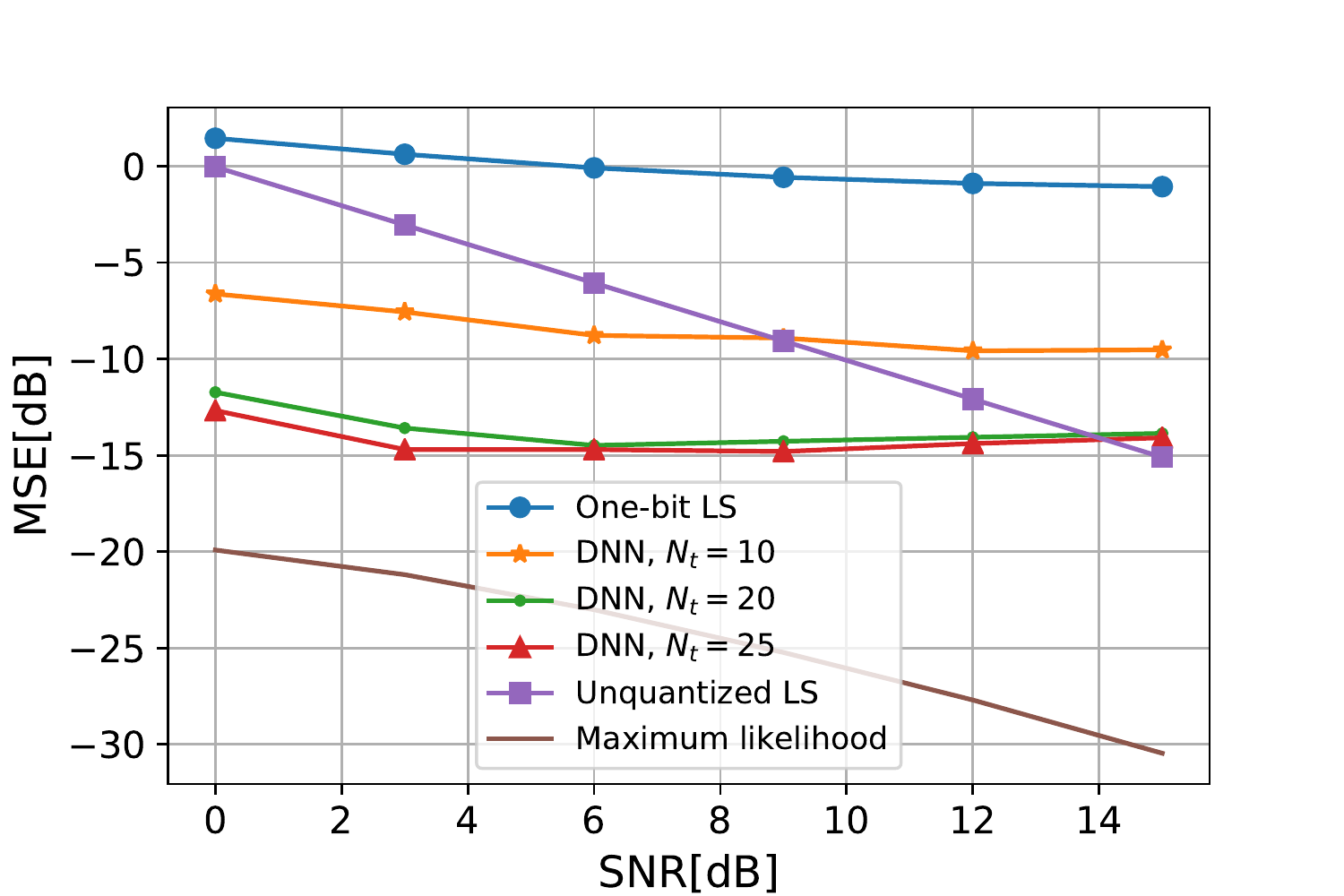}}
\caption{The MSE of the proposed generative supervised DNN model for channel estimation in comparison to LS.}
\end{figure*}
The comparison for $3$ complex Gaussian channel taps is provided in Fig. \ref{fig:threetaps} in terms of MSE including the DNNs trained with $3$ different number of pilots. Note that LS channel estimation can nearly give the same performance whether the number of pilots is $10, 20,$ or $25$. Hence, its performance is only given for $25$ pilots that are sent at the beginning of each coherence interval. As can be seen from this plot, the key parameter that determines the efficiency of the proposed model is the number of pilots. That is, doubling the number of pilots from $10$ to $20$ significantly enhances the performance. Further increase does not have much impact. Hence, it can be deduced that $20$ pilots are reasonable to train an OFDM system that has $64$ subcarriers. The most interesting observation related with Fig. \ref{fig:threetaps} is that although the proposed DNN model have only seen one-bit quantized OFDM samples, it can beat the LS estimation that works with unquantized OFDM samples up to $12$ dB SNR. Additionally, the DNN is always better than the LS channel estimation with one-bit ADCs over all SNRs. On the other hand, there is still room to improve to achieve the maximum likelihood channel estimation performance, which uses the unquantized samples for channel estimation and has significantly higher complexity than the proposed model. Specifically, maximum likelihood channel estimation has exponential complexity $\mathcal{O}(S^N)$ assuming the channel is one of the $S$ states, whereas the proposed model has complexity $\mathcal{O}(W^2)$ such that $W=32N^2$ according to Table \ref{tab:SLM}.

To observe the impact of the number of channel taps to the aforementioned model, the number of channel taps has been increased to $10$ while keeping all the parameters same. This case is depicted in Fig. \ref{fig:tentaps}. It is worth emphasizing that an increase in the number of channel taps leads to a slight improvement in the performance of the proposed DNN. That is, our model is better than the LS channel estimation for unquantized OFDM samples up to $14$ dB.

A natural question is the performance of the model when there are more subcarriers, such as $1024$. One of our empirical observations is that increasing the number of subcarriers significantly increases the complexity, and thus simulation time. This makes sense, because as seen in Table \ref{tab:SLM}, the number of parameters increases quadratically with the number of subcarriers. By this is meant that it is not a reasonable approach to simply increase the dimension of the proposed model to estimate the channel for higher number of subcarriers. To address this issue, large OFDM blocks have to be divided into smaller subblocks, and processed with kernels, which is left to future work. This can be seen as a type of convolution operation.

\subsection{Data Detection}
AE-OFDM architecture can be obtained by implementing the layers  $l_1$-$l_8$ in Fig. \ref{fig:E2E} as tensors, whose parameters are trained through gradient descent with the Adam optimizer. The performance of the proposed AE-OFDM is compared with the conventional uncoded OFDM communication both for unquantized and one-bit quantized samples that employs subcarrier basis detection, i.e., detecting the symbols according to the minimum Euclidean distance criterion after applying a single tap equalization. In particular, a theoretical benchmark error rate is obtained for the ideal unquantized OFDM for Rayleigh fading channels to see the efficiency of the AE-OFDM. For data detection, it is considered that there are $48$ data, $4$ pilot and $12$ guard subcarriers. The CP length is taken as $16$ without loss of any generality. To observe the efficiency of learning in high dimensions, the error rate of the AE-OFDM is presented for different values of $G$, namely for $1, 2, 4$. 

It may be expected that any deep learning based detection for one-bit ADCs can give an error performance in between the unquantized and one-bit quantized OFDM detection. However, this is not the case as demonstrated in Fig. \ref{fig:dataDet}. Specifically, AE-OFDM leads to a slight performance decrease with respect to the one-bit quantized OFDM for $G=1$. On the other hand, there is a performance boost if $G$ is doubled such that we can achieve a BER that is competitive with unquantized OFDM up to $6$ dB. More interestingly, AE-OFDM can beat the theoretical uncoded OFDM error rate in Rayleigh fading channels for $G=4$ up to $10$ dB. This gain resembles the gains seen from channel coding, which consume bandwidth, unlike time domain oversampling. It appears that AE-OFDM is an appealing alternative receiver architecture for low-to-medium SNRs. Note that the BER values are on the order of $0.01$ as expected for uncoded OFDM or any uncoded system in fading. However, the proposed model is flexible enough to be integrated with known coding schemes to yield much lower BER values, and this is left to future work.
\begin{figure} [!h] 
\centering 
\includegraphics [width=3.5in, angle = 0]{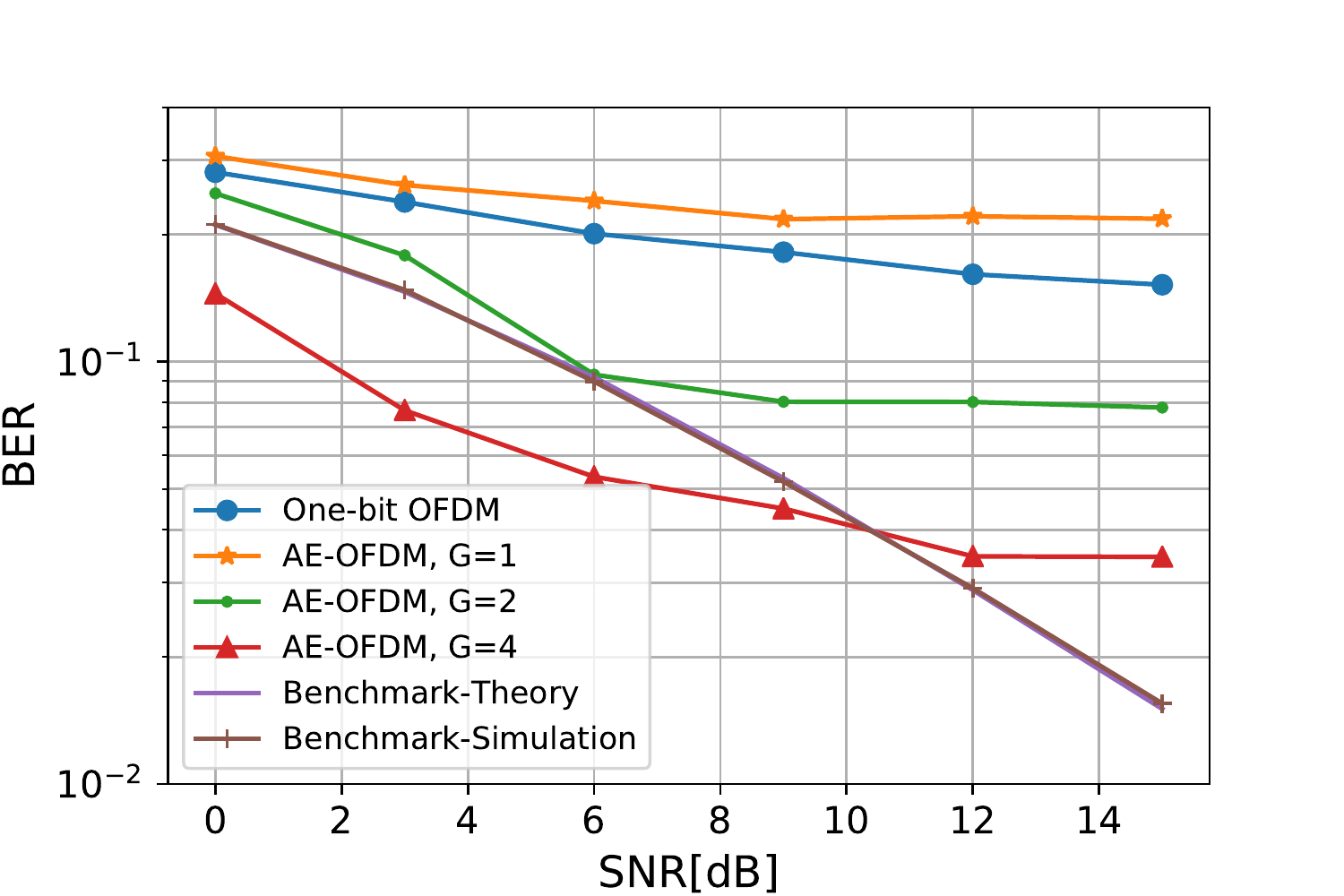}
\caption{The average BER in fading when there are $64$ subcarriers, each of which has been modulated with QPSK. The benchmark is presented both for the average theoretical BER of QPSK in Rayleigh fading and its simulation.}\label{fig:dataDet}
\end{figure}

\section{Conclusions and Future Work} \label{Conclusions}
Replacing the high resolution ADCs with one-bit ADCs can enable a large decrease in receiver cost and power consumption, but leads to a significant performance loss in single antenna OFDM receivers in terms of both channel estimation and data detection if conventional methods are utilized. This paper developed novel deep learning methods for OFDM systems for a moderate number of subcarriers. We proposed a generative supervised DNN for channel estimation using generative modeling and multi-layer neural networks. Our results reveal that reliable channel estimation can be achieved despite the nonlinear impairments of one-bit quantization. Additionally, we proposed an unsupervised autoencoder detection method for OFDM receivers equipped with one-bit ADCs. This model can achieve a satisfactory error rate when the number of neurons in the hidden layers before the quantization layer is sufficiently increased. Promisingly, our results demonstrate that unquantized OFDM performance can be beaten by deep learning methods. 

As future work, it would be interesting to generalize this work to more subcarriers. It is important to emphasize that processing the overall OFDM block with a fully connected neural layer is probably not a reasonable approach for $N\gg64$, and so developing a modified architecture would be necessary. It would be useful to consider more than $1$ transmit and/or receive antenna along with possible MIMO transceiver architectures. In particular, the proposed architectures can in principle be generalized for MIMO communication. However, an efficient method is needed to estimate the channel between each pair of transmit and receive antennas, since this could significantly increase the total number of required pilot symbols in the coherence time interval. Further, the received signals from multiple antennas have to be efficiently combined. Additionally, we have not considered initial acquisition in this paper, which includes (imperfect) time and frequency synchronization, which would be particularly challenging with low resolution quantization.

\appendix[Proof of Theorem \ref{Theorem}] \label{appendix}
Expanding $E[\textbf{F}\mathcal{Q}(\textbf{y}_p)\textbf{s}_p^H]$ using (\ref{quant_sig}) results in 
\begin{equation} \label{laststep}
\begin{split}
E[\textbf{F}\mathcal{Q}(\textbf{y}_p)\textbf{s}_p^H]  & = \textbf{F}\textbf{A}E[\textbf{y}_p\textbf{s}_p^H] + \textbf{F}E[\textbf{d}_p\textbf{s}_p^H] \\
&\myeqa \textbf{F}\textbf{A}E[\textbf{y}_p\textbf{s}_p^H] \\
&\myeqb \sigma_{pilots}^2\textbf{F}\textbf{A}\textbf{H}\textbf{F}^H \\    
\end{split}
\end{equation}
where (a) is due to Lemma \ref{Lemma}, (b) is due to 
\begin{equation} 
E[\textbf{y}_p\textbf{s}_p^H] = \sigma_{pilots}^2\textbf{H}\textbf{F}^H.
\end{equation}

Bussgang's theorem states that if the input to the memoryless system $\mathcal{Q}(\cdot)$ is a zero mean Gaussian process, which is the case for $\textbf{y}_p$, the input-output cross-correlation matrix is proportional to the input auto-correlation matrix such that 
\begin{equation}
C_{{\textbf{y}_p}{\textbf{r}_p}}=\textbf{A}C_{{\textbf{y}_p}{\textbf{y}_p}}
\end{equation}
where 
\begin{equation}\label{corr_mat}
\textbf{C}_{{\textbf{y}_p}{\textbf{y}_p}} = E[\textbf{y}_p\textbf{y}_p^H]
\end{equation}
\begin{equation}
\textbf{C}_{{\textbf{y}_p}{\textbf{r}_p}} = E[\textbf{y}_p\textbf{r}_p^H]
\end{equation}
and \textbf{A} is a diagonal matrix and its $k^{th}$ element is 
\begin{equation}\label{derivative}
[\textbf{A}]_{k,k}=E[\mathcal{Q}'(\textbf{y}_{pk})]=E[2\delta(\textbf{y}_{pk})].
\end{equation}
Using Gaussian probability distribution function in (\ref{derivative}) results in \cite{Papoulis}
\begin{equation}\label{A}
\textbf{A} =\sqrt{\frac{2}{\pi}{\!}}\ \left[\text{diag}(\textbf{C}_{{\textbf{y}_p}{\textbf{y}_p}})\right]^{-\frac{1}{2}}
\end{equation}
where \text{diag($\textbf{C}_{{\textbf{y}_p}{\textbf{y}_p}}$)} refers to the diagonal matrix composed of the diagonal terms of $\textbf{C}_{{\textbf{y}_p}{\textbf{y}_p}}$. 

Expressing (\ref{corr_mat}) as
\begin{equation}
\textbf{C}_{{\textbf{y}_p}{\textbf{y}_p}} = E[\textbf{H}\textbf{F}^H\textbf{s}_p\textbf{s}_p^H\textbf{F}\textbf{H}^H]+\sigma_n^2\textbf{I}_N
\end{equation}
which is equal to
\begin{equation}
\textbf{C}_{{\textbf{y}_p}{\textbf{y}_p}} = E[\textbf{F}^H\boldsymbol{\Lambda}\textbf{s}_p\textbf{s}_p^H\boldsymbol{\Lambda}^H\textbf{F}]+\sigma_n^2\textbf{I}_N
\end{equation}
leads to
\begin{equation}\label{diagCyp}
\text{diag($\textbf{C}_{{\textbf{y}_p}{\textbf{y}_p}}$)}=(\sigma_{chn}^2\sigma_{pilots}^2+\sigma_n^2)\textbf{I}_N.
\end{equation}
Substituting (\ref{diagCyp}) in (\ref{A}) produces
\begin{equation}\label{Aderived}
\textbf{A} =\sqrt{\frac{2}{\pi(\sigma_{chn}^2\sigma_{pilots}^2+\sigma_n^2)}}\ \textbf{I}_{N}.
\end{equation}
Using (\ref{Aderived}) in (\ref{laststep}) along with (\ref{circ}) completes the proof.

\end{document}